\documentclass[10pt,a4paper,reqno]{article}
\usepackage{amsfonts}
\usepackage{natbib}
\usepackage{url}

\usepackage{amssymb}
\usepackage{amsmath}
\usepackage{amsthm}

\usepackage{graphics}
\usepackage{graphicx}
\usepackage{hyperref}

\newcommand{\re}{{\mathbb R}}    % real numbers
\newcommand{\za}{{\mathbb Z}}    % integer numbers
\newcommand{\var}[1]{{\mathbb Var}(#1)}     % variance
\newcommand{\covar}[2]{{\mathbb Cov}(#1,#2)}     % covariance
\newcommand{\corr}[2]{{\mathbb Cor}(#1,#2)}     % correlation
\newcommand{\X}{\mathcal{X}} % vector X
\newcommand{\Y}{\mathcal{Y}} % vector Y
\newcommand{\Z}{\mathcal{Z}} % vector Z
\newcommand{\Ztilde}{\mathcal{\tilde{Z}}} % vector Z tilde
\newcommand{\EE}{\mathcal{E}} % vector E
\newcommand{\ARMA}[2]{\text{ARMA}(#1,#2)}
\newcommand{\AR}[1]{\text{AR}(#1)}
\newcommand{\MA}[1]{\text{MA}(#1)}
\newcommand{\VARMA}[2]{\text{VARMA}(#1,#2)}
\newcommand{\VAR}[1]{\text{VAR}(#1)}
\newcommand{\VMA}[1]{\text{VMA}(#1)}
\newcommand{\WN}[2]{\text{WN}(#1,#2)}

\newcommand{\ind}{\mathbf{1}}
\newtheorem{theorem}{Theorem}
\newtheorem{obs}{Observation}
\newtheorem{prop}[theorem]{Proposition}

\newtheorem{corollary}[theorem]{Corollary}

\theoremstyle{definition}
\newtheorem{definition}[theorem]{Definition}

\theoremstyle{remark}
\newtheorem{remark}[theorem]{Remark}

\newtheorem{example}{Example}

\numberwithin{theorem}{section}
\numberwithin{equation}{section}

\title{Scaling portfolio volatility and calculating risk contributions in the presence of serial cross-correlations}

\author{Nikolaus Rab \thanks{Vienna University of Technology, Financial and Actuarial Mathematics, Wiedner Hauptstra\ss{}e 8-10/105-1, 1040 Vienna, Austria, e-mail: \url{nikolaus.rab@student.tuwien.ac.at}}
\and 
				Richard Warnung \thanks{Risikomanagement, Raiffeisen Kapitalanlage-Gesellschaft~m.~b.~H., Schwarzenbergplatz~3, 1010 Vienna, Austria, and external lecturer at Vienna University of Technology, Financial and Actuarial Mathematics, e-mail: \url{richard.warnung@rcm.at} resp. \url{rwarnung@gmx.at}. The contents of this paper
are the authors' sole responsibility.
They do not necessarily represent the views of Raiffeisen Kapitalanlage-Gesellschaft~m.~b.~H.}
\thanks{Both authors thank the members of the risk management department of Raiffeisen Kapitalanlage-Gesellschaft~m.~b.~H. for fruitful discussions. Personal communcations with Paul Gilbert, the author of the DSE package(~\cite{Gilbertdse}), is gratefully acknowledged.}
}

\begin{document}
\maketitle 
\begin{abstract}
In practice daily volatility of portfolio returns is transformed to longer holding periods by multiplying by the square-root of time which assumes that returns are not serially correlated. Under this assumption this procedure of scaling can also be applied to contributions to volatility of the assets in the portfolio.
Close prices are often used to calculate the profit and loss of a portfolio.
Trading at exchanges located in distant time zones this can lead to significant serial cross-correlations of the closing-time returns of the assets in the portfolio. These serial correlations cause the square-root-of-time rule to fail. Moreover volatility contributions in this setting turn out to be misleading due to non-synchronous correlations.
We address this issue and provide alternative procedures for scaling volatility and calculating risk contributions for arbitrary holding periods. 
\end{abstract}

% for amsart

%\keywords{Euler-allocation, square-root of time rule, volatility scaling, serial correlation, weakly stationary processes, volatility contributions}
\noindent
Keywords: portfolio market risk, volatility scaling, square-root-of-time rule, Euler-allocation, volatility contributions, serial correlation, weakly stationary processes, Box-Jenkins models, vector arma models.

\begin{center}
\textbf{This is a preprint of an article forthcoming in the Journal of Risk ~\url{http://www.thejournalofrisk.com/}.}
\end{center}

%\section*{Acknowledgments}
%The authors thank the members of the risk management department of Raiffeisen Kapitalanlage-Gesellschaft for fruitful discussions.

\section{Introduction and Motivation}
In this article we consider a portfolio consisting of $n>0$ assets whose returns at time $t\in \za$ are modelled by the random vector $\X_t := (X_t^1,\ldots,X_t^n)^T$ for $t \in \za$, where $\za$ denotes the integers. We will use methods from time series analysis where $\za$ is the natural range for the time index (see for example~\cite{mills:timeseries}).
The percentage weights of the assets in the portfolio are denoted by the vector $\lambda:=(\lambda_1,\ldots,\lambda_n)^T \in \re^n$. Thus we allow for short-selling as well as leverage where the leverage can be measured by $\sum_{i=1}^n \lambda_i$. 

Building a portfolio using the weights $\lambda$ the random portfolio return at time 
$t\in \za$, denoted by $X_t(\lambda)$, is given by
\begin{align}\label{eq:pfReturn}
	X_t(\lambda) := \lambda^T \X_t = \sum_{i=1}^n \lambda_i X_t^i, \quad \text{for } t \in \za.
\end{align}
For our analysis we will assume that $(\X_t)_{t \in \za}$ and therefore, as we show in Proposition~\ref{prop:multi.var.pf.stationary}, also $(X_t(\lambda))_{t \in \za}$ is a weakly stationary multivariate respectively univariate time series. Thus mean and covariances exist and covariances depend on the lag and not the absolute %time points (for details see for example~\cite[Chapter 4, Section 1]{McNeil/Frey/Embrechts:RiskManagement}).
time points (for details see for example~\cite{TimeSeries:BoxJenkins,TimeSeries:BrockwellDavis}).

In practice it is assumed that the component time series in $(\X_t)_{t \in \za}$ do not exhibit any serial correlation. Thus it is usually assumed that yesterday's return of asset $i$ is not correlated with today's return of asset $j$. But there are situations where this assumption is clearly not fulfilled as we will see in the following.

Holding a globally diversified portfolio of stock-index futures and analyzing the returns between close prices one sometimes experiences significant correlations in lagged returns. This may be due to non-overlapping trading times of various exchanges (e.g. New York and Tokyo). For a detailed analysis of the correlation bias we refer to~\cite{Kahya:Correlation,Coleman:EstimatingCorrelation}. This problem can not be addressed by taking prices from a point in time where all exchanges involved are open as such a moment might simply not exist. This lead-lag effect is a prime example for the application of multivariate time series analysis, see for example~\cite[Chapter 4, Section 5]{McNeil/Frey/Embrechts:RiskManagement} and~\cite{deJongHighFrequency} for empirical aspects in high frequency. %In the %following we give a concrete example and motivate our focus on vector-autoregressive models and especially the %$\VAR{1}$ case.

Modelling prices in different time zones we consider days  $t \in \za$ and closing-time fractions $x_1 \le x_2 \le \cdots \le x_n$ of the $n$ exchanges where the assets in our portfolio are traded. By this we mean that the first market closes each day at the point in time $x_1$ (e.g. if $x_1=1/3$ then we mean 08:00~CET) and so forth. For the return series of a portfolio $(X_t(\lambda))_{t \in \za}$ we have to define the following notion\footnote{We thank an anonymous referee for pointing out that the following clarification of notions is essential.}:

\begin{definition}[Closing-time return]\label{defi:closingtimereturn}
Calculating the return of an investment using~\eqref{eq:pfReturn}
where each return $X_t^i,i=1,\ldots,n$ is calculated from the close price in the respective market $P_t^i $,i.e.
\begin{align}\label{eq:closingtimeasset}
	X_t^i = (P_{t+x_i}^i - P_{t-1+x_i}^i)/P_{t-1+x_i},\quad\text{for } t \in \za,
\end{align}
where $x_i$ is the closing-time fraction, is called closing-time return.
\end{definition}

Thus we use the notion of closing-time return as the return on investment that comes from book-keeping the prices each day at the point in time when the respective market closes. This is e.g. 08:00~a.m.~CET for Japanese stocks and 10:00~p.m.~CET for stocks traded in the US. This return is often used by accountants to calculate the performance of an investment fund or the profit and loss of a trader. 

As one can see from the example of Japan and USA, a closing-time return does not include the reaction of markets to events that happen when they are closed. Economically an event that affects the US market while the Japanese market is closed will potentially affect Japan one day later. Profits and losses accounted for by~\eqref{eq:closingtimeasset} affect the aggregate performance of the portfolio with some lag as profits and losses can be increased or decreased by the reaction of the Japanese market on the following day. 

This lagging reaction changes our view of the risk that we economically have. The idealized concept of a return that takes into account lagging reactions on one and the same day is called \emph{contemporaneous return} in this article as opposed to the non-contemporaneous closing-time return. This return could be thought of as a result of shifting all trading times, i.e. opening and closing-times, such that trading takes place contemporaneously. Note that the trading times in the US and Japan do not even overlap. Thus the calculation of the risk in contemporaneous returns can be seen as an estimation problem that can for example be tackled by the Newey-West estimator (see~\cite{NeweyWest}). We briefly show an application of this estimator in the end of Example~\ref{ex:portfolio} below.

Summing up we have a closing-time return used by accountants to asses the performance of a fund or a trader in a straight and clear way. On the other hand we have the contemporaneous return that tries to reflect the economic interaction of markets. In this paper we mainly deal with closing-time returns and provide formulas to calculate genuine risk and risk contributions.
In the following we give a concrete example and motivate our focus on vector-autoregressive models and especially the $\VMA{1}$ case.

\begin{example}\label{ex:portfolio}
Consider a portfolio of stock index futures traded in Japan, China, Europe, South Africa and the United States. Note that weights sum up to $105\%$ which corresponds to a small degree of leverage. The following table shows the names of stock indices in these regions, example weights and the closing-times of the exchanges in CET. For such a portfolio there is no point in time where all exchanges are open. Using close prices at all exchanges one has to treat the lagged correlations between the returns of different assets as it is illustrated below.
\begin{table}[!htb]
	\begin{center}
		\begin{tabular}{l l|c c c}
			& Asset & Currency & Exposure &Closing (CET)  \\ \hline
1 &           Topix (TOP) &  JPY   &15\%  &  08:00 a.m.  \\
2 &        H-shares (HSH) &  HKD &  15\% &  10:00 a.m.  \\
3 & DJ Euro Stoxx 50 (DJS) & EUR  &  15\%  & 05:30 p.m.  \\
4  &      Swiss Market (SWI)&  CHF   & 15\%  & 05:30 p.m.     \\ 
5 &     JSE TOP 40 (JSE)&  ZAR   & 15\%   & 06:30 p.m.   \\  
6 &  Russell 2000  (RUS)&  USD    &15\%    & 10:00 p.m.        \\
7 &      NASDAQ 100 (NAS)&  USD   &15\%  & 10:00 p.m.   
		\end{tabular}
	\end{center}
	%	\caption{A global sample portfolio}
		\label{tab:21}
\end{table}

The correlations of the returns\footnote{Estimated from a 250 days window ending in April 2010.} of the indices underlying these futures contracts are plotted in Figure~\ref{fig:LagZero} where the numbering corresponds to the numbering in the table above. At lag zero we see that geographically close markets have relatively strong correlation in returns. One recognizes an Asian (TOP,HSH), a European (DJS,SWI) (also showing significant correlations with South Africa) and an American block (RUS,NAS). In Figure~\ref{fig:LagOne} the correlations of the returns (y-axis, \emph{today}) to lagged returns (x-axis, \emph{yesterday}) are plotted. The auto-correlations on the main diagonal (i.e. correlations of returns of one asset with lagged returns of the same asset) are not significant which is consistent with a random walk. But on the other hand this plot illustrates that, in our example, returns in Asian markets have significant correlation with lagged returns of the European as well as the US block. This illustrates the lead-lag effect. We will come back to this portfolio in Section~\ref{sec:time.series} and apply advanced procedures to analyze risk.
% new part

We continue with a rather heuristic analysis of the situation in the spirit of ~\cite{Kahya:Correlation,Coleman:EstimatingCorrelation} and~\cite{Bergomi:AsynchronousMarkets}.
Let 
\[
(\tilde{B}_t)_{t \in \re} = (\tilde{B}_t^1,\ldots,\tilde{B}_t^n)_{t \in \re}
\]
be an $n-$dimensional Brownian motion with the identity matrix as correlation matrix, i.e.
\begin{align*}
\covar{\tilde{B}_t^i}{\tilde{B}_t^j} = 0, \quad \text{if }i \neq j.
\end{align*}
 Then, with $\Sigma=(\sigma_{i,j})_{i,j =1}^n$  an $n \times n$ positive-semidefinite matrix and $ \Sigma^{1/2}$ its Cholesky-decomposition, it is clear that $(\tilde{B}_t)_{t \in \re}$ defined by $B_t = \Sigma^{1/2} \cdot \tilde{B}_t$ for $t \in \re$ is a Brownian motion with covariance matrix 
$\Sigma$, meaning that 
\begin{align}\label{eq:originalsigma}
	\covar{B_t^i}{B_t^j} = t \sigma_{i,j}, \quad \text{for } t > 0 \text{ and } i,j=1,\ldots,n.
\end{align}
Then the increments $B_t-B_s$ have a multivariate normal distribution with covariance matrix $(t-s) \Sigma$.
Recording prices at different closing-times during the day can be described by observing prices at integer times $t$ shifted by the closing-time fraction $x_i,i=1,\ldots n$. Thus we model our closing-time returns from Definition~\ref{defi:closingtimereturn} by the vector
\[
   \X_t = (X_t^1,\ldots,X_t^n) = (B^1_{t+x_1} - B^1_{t+x_1-1},\ldots,B^n_{t+x_n} - B^n_{t+x_n-1}), \quad \text{for } t \in \za, 
\]
where we assume that assets are ordered with respect to their closing-times, i.e. $x_1 \le x_2 \le \cdots \le x_n$.
If $x_1 = x_2 \dots = x_n$ then we can shift time by $x_1$ by stationarity of Brownian motion and get back the standard model with $\Sigma$ as the covariance matrix of the asset returns.
Assuming without loss of generality that $x_i>x_j$, i.e. that market $i$ closes later than market $j$ then we get for the general case
\begin{align*}
	\covar{X_t^i}{X_t^j} &= \covar{B^i_{t+x_i} - B^i_{t-1+x_i}}{B^j_{t+x_j} - B^j_{t-1+x_j}} \nonumber  \\ 
&= \sigma_{i,j}(1-(x_i-x_j)),
\end{align*}
where we have used
\[
	[t-1+x_j,t+x_i] = [t-1+x_j,t-1+x_i) \cup  [t-1+x_i,t+x_j) \cup [t+x_j,t+x_i] 
\]
as a decomposition for the time interval and
\begin{align*}
B^i_{t+x_i} - B^i_{t-1+x_i} &= \underbrace{B^i_{t+x_j} - B^i_{t-1+x_i}}_{B^i_t \text{ on } [t-1+x_i,t+x_j)} + B^i_{t+x_i} - B^i_{t+x_j}, \quad \text{ respectively, }\\
B^j_{t+x_j} - B^j_{t-1+x_j} &= \underbrace{B^j_{t+x_j} - B^j_{t-1+x_i}}_{B^j_t \text{ on } [t-1+x_i,t+x_j) } + B^j_{t-1+x_i} - B^j_{t-1+x_j},
\end{align*}
for the Brownian motion.
Thus the covariance matrix of the closing-time returns $\X_t$ is given by 
\begin{align}\label{eq:intro_lag0}
	\Sigma_X := \left( \sigma_{i,j}(1-|x_i-x_j|)  \right)_{i,j=1}^n.
\end{align}
Note that on the main diagonal we have $ \sigma_{i,i} = \sigma^2_i$, i.e. variances are not biased but covariances are.

Furthermore we consider the lag one covariance matrix of $\X_t$. Let 
\[
	\Gamma(1) = \covar{\X_{t+1}}{\X_{t}}
\]
such that
\begin{align*}
	\Gamma(1)_{i,j} = \covar{X_{t+1}^i}{X_{t}^j}, \quad \text{for } i,j=1,\ldots,n.
\end{align*}
 Then, assuming $x_j>x_i$ , due to the overlap of the  interval $[t-1+x_j,t+x_j]$ where we observe $X^j_t$, respectively $[t+x_i,t+1+x_i]$ where we observe $X^i_{t+1}$, we have
\begin{align}\label{eq:intro_lag1}
	\covar{X_{t+1}^i}{X_{t}^j} =
	\begin{cases}
	0, &\text{ if }  x_j \le x_i,       \text{ and,}\\
	(x_j-x_i)\sigma_{i,j}, &\text{ if } x_j > x_i. 
	\end{cases}
\end{align}
With the assumed ordering of the markets $\Gamma(1)$ is  an upper triangular matrix with zero main diagonal.
The covariance matrices of lag two and greater are zero as there is no overlap in the time intervals of interest anymore.
For example at lag two consider that
\[
	\covar{X_{t+2}^i}{X_{t}^j} = \covar{B_{t+2+x_i}^i-B_{t+1+x_i}^i}{B_{t+x_j}^j-B_{t+x_j-1}^j}  = 0,
\]
as 
\[
	 [t-1+x_j,t+x_j]  \cap [t+1+x_i,t+2+x_i]   = \emptyset.
\]

Finally note that the vector $(\X_{t+1},\X_{t})$  is jointly Gaussian and the conditional law of $\X_{t+1}$ given $\X_t$ is Gaussian with mean $\Gamma(1)\Sigma_X^{-1} \X_t$ and covariance matrix 
\[
	\Sigma_X - \Gamma(1)\Sigma_X^{-1}\Gamma(1)^T.
\] 
Thus knowing the conditional mean of $\X_{t+1}$ given $\X_t$  we can write
\begin{align}\label{eq.wrongVAR}
	\X_{t+1} = \Phi \X_t + \EE_{t+1} , \quad \text{for } t \in \za, 
\end{align}
with $\Phi = \Gamma(1)\Sigma_X^{-1}$ and a process $ (\EE_{t})_{t\in\za} $ which is not a white noise process as some lines of calculations show. Although the (visual) form~\eqref{eq.wrongVAR} reminds of a $\VAR{1}$ model, the vanishing lag covariances greater than one are inconsistent with a $\VAR{1}$-model (see the form of the lagged covariance matrices in~\eqref{eq:autovoar.functions} in Section~\ref{sec:time.series}).
We rather assume a $\VMA{1}$-model of the form
\begin{align}\label{eq:MAAssets}
	\X_{t+1} = \Theta \Z_t + \Z_{t+1} , \quad \text{for } t \in \za, 
\end{align}
with a white noise process $ (\Z_{t})_{t\in\za} $. We will come back to this model class in Section~\ref{sec:time.series}.

%This observation motivates our focus on vector-autoregressive models in general and $\VAR{1}$-models for the problem at %hand when we consider closing-time returns.

As mentioned before calculating the covariance matrix of contemporaneous returns leads to a special estimation technique and we give two examples. 
It is straight forward to define an estimator by
\begin{align}\label{eq:NWintro}
	\tilde{\Sigma} = \underbrace{\Sigma_X }_{:= \Gamma(0)} + \Gamma(1) + \Gamma(1)^T.
\end{align}
It is known that the estimator~\eqref{eq:NWintro} does not always yield a positive-semidefinite covariance matrix (see~\cite{NeweyWest}).
However, considering Equations~\eqref{eq:intro_lag0} and~\eqref{eq:intro_lag1} we see that this covariance estimator of contemporaneous returns, $\tilde{\Sigma}$, equals $\Sigma$ asymptotically, which is the covariance matrix of the Brownian motion defined in~\eqref{eq:originalsigma}. Thus in this setting the  estimator~\eqref{eq:NWintro} reveals the true contemporaneous covariance matrix.
For the general case there is an estimator which always yields a  positve-semidefinite covariance matrix, the so called Newey-West estimator. The estimator introduced in~\cite{NeweyWest} up to lag 1, denoted by $\Sigma_{NW}$, is given by
\begin{align}\label{eq:NWintro_real}
	\Sigma_{NW} = \Gamma(0) + \frac12 \left(\Gamma(1) + \Gamma(1)^T  \right).
\end{align}
However in the setting of non-contemporaneous trading we prefer~\eqref{eq:NWintro}.
\end{example}

\begin{figure}[htbp]
	\begin{center}
	\includegraphics[width=0.8\textwidth]{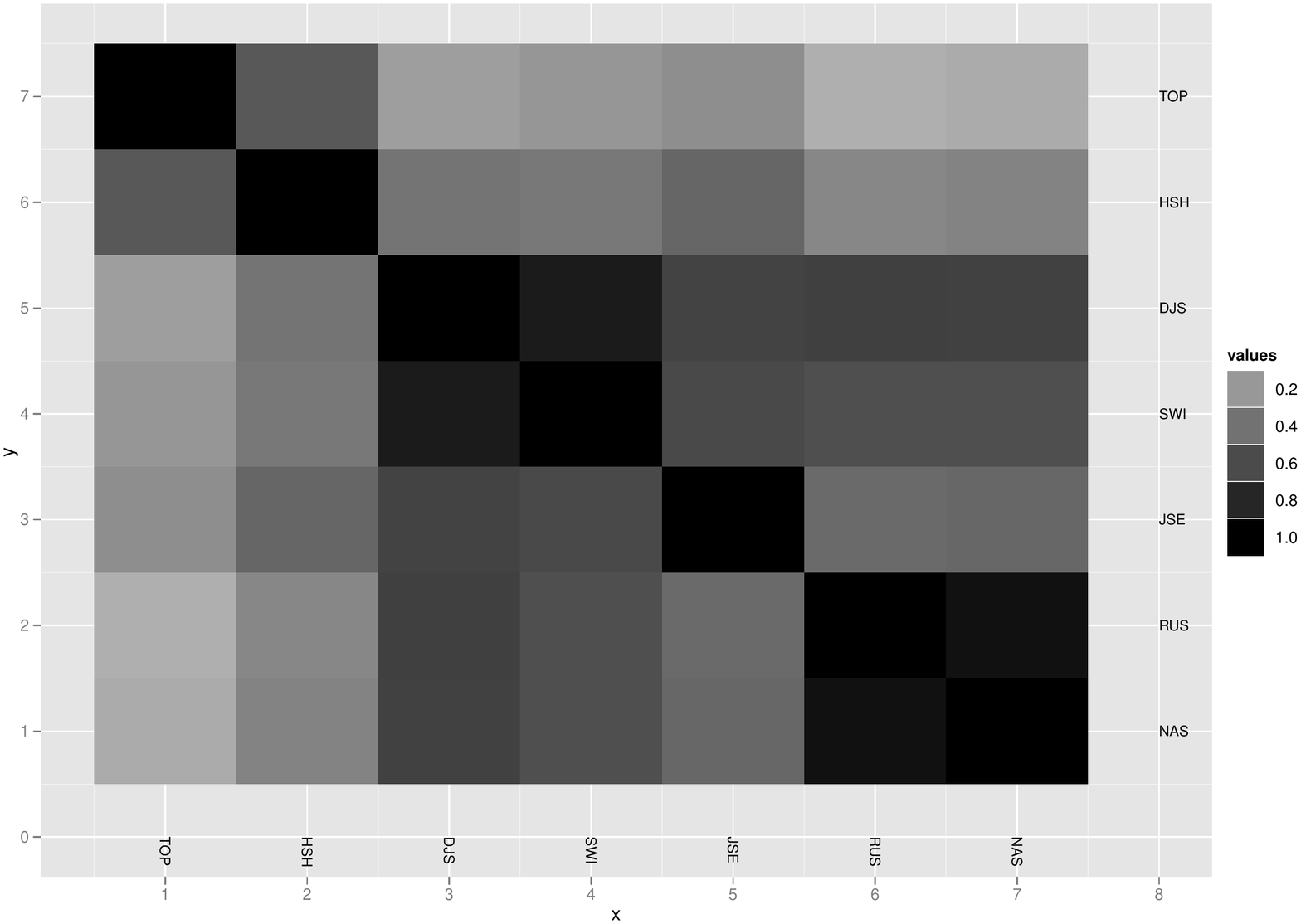} 
	\end{center}
	\label{fig:LagZero} 
	\caption{Heatmap of correlations of lag zero.}
\end{figure}

\begin{figure}[htbp]
	\begin{center}
	\includegraphics[width=0.8\textwidth]{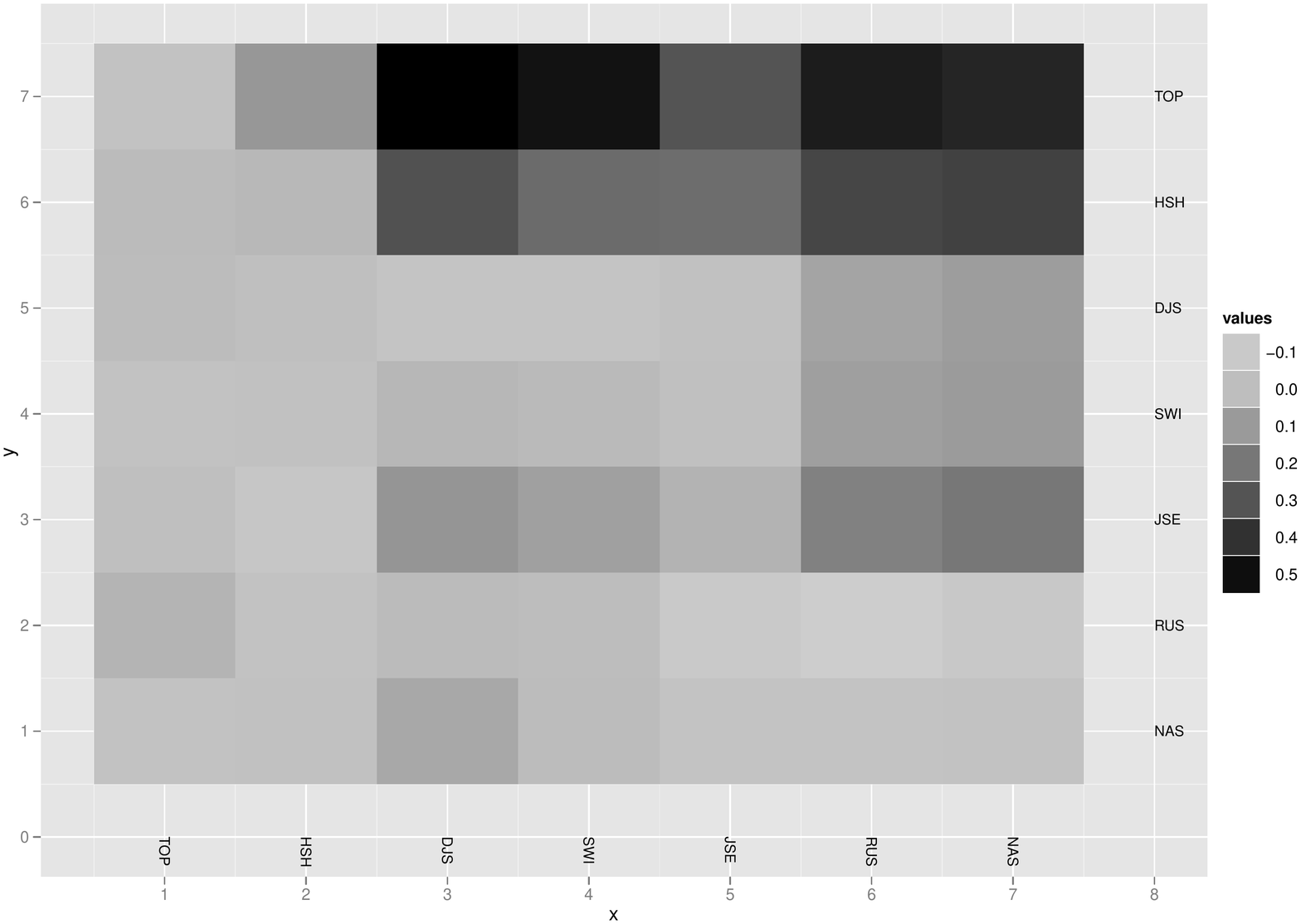}
	%\caption{Heatmap of correlations of returns at lag one along the x-axis and returns at lag zero along the y-axis.}
		\end{center}
	\label{fig:LagOne}
 \caption{Heatmap of correlations of  lag one.}
\end{figure}

In the following we mainly consider closing-time returns and do not always mention this fact. If we make some remarks concerning contemporaneous returns this will be mentioned explicitly.

The risk of the return $X_t(\lambda)$ holding the portfolio for one day is often measured by volatility which is simply the square-root of variance. Stressing that this volatility depends on the asset weights $\lambda$, we write
\begin{align}\label{eq:vola}
	\sigma(\lambda):= \sqrt{\var{X_t(\lambda)}} = \sqrt{\var{\sum_{i=1}^n \lambda_i X_t^i}},\quad \text{for } t \in \za.
\end{align}

\begin{remark}
If we furthermore assume that the vector $\X_t$ is normally distributed for each $t$, then knowing the volatility and the mean we can calculate value-at-risk and expected shortfall by the well-known formulas~\cite[Chapter 2]{McNeil/Frey/Embrechts:RiskManagement}.
\end{remark}

\begin{remark}\label{rem:matrixvola}
If $\Sigma = (\Sigma_{i,j})_{i,j=1}^n$ denotes the covariance matrix of $\X_t$, i.e. $\Sigma_{i,j} = \covar{X_t^i}{X_t^j}$, then it is well known that 
\[
	\sigma(\lambda) = \sqrt{\lambda^T \Sigma \lambda}.
\]

\end{remark}

Regulatory rules require the calculation of the risk of a portfolio for a holding period of $d\ge1$ days (e.g. 10 or 20 days) and it is common to quote volatility as per annum (i.e. $d=250$ or $d= 252$).
For the clearness of presentation we will assume that we model returns at $t=1,\ldots,d$ when considering a holding period of $d$ days.

\begin{definition}[Volatility for a holding period of $d$ days]
We denote the volatility of the portfolio return given in~\eqref{eq:pfReturn} holding the assets for $d$ days by
$\sigma(\lambda,d)$ and thus
\begin{align}\label{eq:voladdays}
	\sigma(\lambda,d) := \sqrt{\var{\sum_{k=1}^d X_{k}(\lambda)}}.
\end{align}	
\end{definition}

In practice it is assumed that the returns $X_{t}(\lambda)$ are serially uncorrelated for $t \in \za$
and that variance is stationary, i.e. that $\var{X_{k}} = \sigma(\lambda)$ for $k=1,\ldots,d$. These assumptions lead to the square-root-of-time rule. 

\begin{prop}[Square-root-of-time rule]
Under the above assumptions it holds that
\begin{align}\label{eq:sqrtrule}
	\sigma(\lambda,d) = \sigma(\lambda)\sqrt{d}, \quad\text{for } d\ge1,
\end{align}	
where $\sigma(\lambda)$ denotes the one-day volatility.
\end{prop}
Note that, due to its simplicity, the square-root-of-time rule is often used in practice. Using it without checking its appropriateness can lead to poor risk estimation (see~\cite[Chapter 2]{McNeil/Frey/Embrechts:RiskManagement} and references therein).

Having calculated the volatility of the portfolio return the Euler allocation is used in practice to define risk contributions by the assets (see~\cite{Tasche00riskcontributions,Tasche2008captialallocationBU} and~\cite[Chapter 6]{McNeil/Frey/Embrechts:RiskManagement}).

\begin{definition}[Contribution to volatility]\label{def:volacontrib}
The contribution to the volatility of the portfolio return $X_t(\lambda)$ of asset $i$ by the Euler rule, denoted by 
	$\sigma_i(\lambda)$, is given by
	\begin{align}\label{eq:varvonstrib}
		\sigma_i(\lambda) := \lambda_i \frac{\covar{X^i_t}{X_t(\lambda)}}{\sqrt{\var{X_t(\lambda)}}}, \quad\text{for } i=1,\ldots,n.
	\end{align}
\end{definition}

\begin{remark}
Using the same notation as in Remark~\ref{rem:matrixvola} one can easily see that
\[
	\sigma_i(\lambda) = \lambda_i \frac{(\Sigma \lambda)_i}{\sigma(\lambda)}, \quad \text{for }i=1,\ldots,n,
\]
where $(\Sigma \lambda)_i$ denotes component $i$ of the vector $\Sigma \lambda$.
\end{remark}

Assuming that asset returns $\X_t$ are not serially correlated for $t \in \za$ we can make the following definition for the risk contribution for a holding period of $d$ days, which we denote
by $\sigma_i(\lambda,d)$ for $i=1,\ldots,n$:

\begin{prop}[Square-root-of-time rule for risk contributions]
Under the above assumptions it holds that
\begin{align}\label{eq:sqrtrule.contrib}
	\sigma_i(\lambda,d) = \sigma_i(\lambda)\sqrt{d},\quad\text{for } i=1,\ldots,n  \text{ and } d\ge 1,
\end{align}	
where $ \sigma_i(\lambda)$ is given in~\eqref{eq:varvonstrib}.
\end{prop}

\begin{remark}
Defining $\sigma(\lambda)$ and $\sigma_i(\lambda),i=1,\ldots,n$ as above it is easily seen that we have full allocation for all holding periods $d$, i.e. 
\[
	\sigma(\lambda,d) = \sum_{i=1}^n \sigma_i(\lambda,d),\quad\text{for } d\ge 1.
\]
Furthermore considering relative risk contributions it is easily seen by the following equation that they do not change when considering longer holding periods:
\[
	\frac{\sigma_i(\lambda,d)}{\sigma(\lambda,d)} = \frac{\sigma_i(\lambda) \sqrt{d}}{\sigma(\lambda)\sqrt{d}} = \frac{\sigma_i(\lambda)}{\sigma(\lambda)}, 
\]
for $i=1,\ldots,n$ and $d \ge 1$.
\end{remark}

The above rules for scaling volatility are well-known and used in practice. In the following section we drop the assumption that asset returns are serially uncorrelated. In Section~\ref{sec:volacontrib}
we derive general formulas for scaling volatility and volatility contributions in a setting with serial correlations. 
Modelling the portfolio return as a univariate process with auto-correlations we propose proper volatility scaling in Subsection~\ref{subsec:volacontrib.uni} while the multivariate setting of Subsection~\ref{subsec:volacontrib.multi} additionally allows for the calculation of proper risk contributions. In Section~\ref{sec:time.series} we come back to the concrete problem of non-contemporaneous trading and apply these results to $\VMA{1}$-models and derive the corresponding formulas.
Then we return to the data of Example~\ref{ex:portfolio} and show that the square-root-of-time rule can seriously underestimate volatility. Furthermore it can give misleading indications of risk contributions. Finally we take a short detour and propose $\VAR{1}$-models to tackle the problem of genuine auto-correlations in the sense of~\cite{Andersonetal:StockReturn}.

While we focus on auto-regressive modelling of the returns, a GARCH-approach to the problem of scaling volatility is given in~\cite{RePEc:wop:pennin:97-34} and the scaling in a model with jumps in returns is considered in~\cite{RePEc:eee:jbfina:v:30:y:2006:i:10:p:2701-2713}. For a study on the square-root-of-time rule for tail risk we refer to the recent paper~\cite{WangYehCheng:HowAccurateSROT}. While an analysis of temporal aggregation of ARMA-models with GARCH errors is given in~\cite{DrostNijman:TermporalAggregation} we focus on risk contributions.  To our knowledge the question of scaling volatility contributions is so far not covered in the literature.

\section{Volatility Contributions under Serial Correlation}
\label{sec:volacontrib}
Recall that we analyze volatility scaling and volatility contributions in the setting of weakly stationary processes. 
There are circumstances where the assets that constitute the portfolio are not known and one can only model volatility on the level of portfolio returns. In this setting we will propose a scaling rule that takes auto-correlations into account. Later on
we analyze the situation when assets are known. Then a bottom-up modelling of the portfolio is possible and the whole covariance structure of lagged asset returns can be estimated.

\subsection{The univariate model with auto-correlations}\label{subsec:volacontrib.uni}
First we make the following definition for the auto-covariance and auto-correlation function of portfolio returns $(X_t(\lambda))_{t \in \za}$:

\begin{definition}[Auto-covariances of a univariate time series]\label{defi:univariate.covar}
Let $(X_t(\lambda))_{t \in \za}$ denote a univariate weakly stationary stochastic process, then the auto-\-covariance function and the auto-correlation function is denoted by 
\begin{align}
	\gamma(k) &:=\covar{X_t(\lambda)}{X_s(\lambda)}  \quad \text{and} \label{eq:uni.acf}\\
	\rho(k)&:=\corr{X_t(\lambda)}{X_s(\lambda)} = \gamma(k)/\gamma(0), \label{eq:uni.acorf}
\end{align}
for $t,s \in \za$ where $ |t-s| = k$, respectively.
\end{definition}
In the presence of auto-correlations in the portfolio returns the scaling by the square-root of time is not accurate and the following proposition states the correct scaling in this setting.
\begin{prop}[Volatility for a holding period of $d$ days]\label{prop:scaling.univariate}
Let $(X_t(\lambda))_{t \in \za}$ be a univariate weakly stationary stochastic process with auto-covariance function $\gamma(\cdot)$ as defined in Definition~\ref{defi:univariate.covar} then the volatility of the return when holding the portfolio over $d\ge2$ days is given by
\begin{align}\label{eq:scaling.unvariate}
	   \sigma(\lambda,d) = \sigma(\sum_{k=1}^d X_k(\lambda)) = \sqrt{d \gamma(0) + 2 \sum_{k=1}^{d-1} (d-k) \gamma(k)}.
\end{align}
\end{prop}
The proof of Proposition~\ref{prop:scaling.univariate} is given in the appendix and the correction to scaling by the square-root of time is clearly seen in the following corollary.
\begin{corollary}[Scaling rule for the univariate model]\label{cor:scaling.univariate}
Let $(X_t(\lambda))_{t \in \za}$ as in Proposition~\ref{prop:scaling.univariate} and $\rho(\cdot)$ be its auto-correlation function then
\begin{align}\label{eq:scaling.unvariate.rule}
	\sigma(\lambda,d) = \sigma(\lambda) \delta(d),
\end{align}
where $\sigma(\lambda)$ is the one-day volatility from~\eqref{eq:vola} and the factor $\delta(d)$ is given by
\begin{align}\label{eq:scaling.unvariate.constant}
	\delta(d) = \sqrt{d + 2 \sum_{k=1}^{d-1} (d-k) \rho(k)}, \quad \text{for } d\ge2.
\end{align}
\end{corollary}
The short proof of this corollary can also be found in the appendix. Considering~\eqref{eq:scaling.unvariate.constant} we see that in our univariate model the scaling is given by the square-root of time corrected by an expression taking into account all relevant auto-correlations. If these are zero then~\eqref{eq:scaling.unvariate.constant} reduces to the well-known formula~\eqref{eq:sqrtrule}.

%The next proposition gives a general estimate of the error when the square-root-of-time rule is applied for scaling volatility %when it is not appropriate.
The next corollary gives a crude estimate of how much higher the true scaling factor can be compared to the square-root-of-time rule. %This tells us how wrong the square-root-of-time rule can be if it is not appropriate.

\begin{corollary}[Error when using the square-root-of-time rule]\label{cor:square.root.error}
Let $(X_t(\lambda))_{t \in \za}$ as in Proposition~\ref{prop:scaling.univariate} and $\delta(d)$ be the correct scaling factor, then
the proportion of correct scaling to an application of the square-root-of-time rule can be estimated as
\begin{align}\label{eq:square.root.error}
	\frac{\delta(d)}{\sqrt{d}} \le \sqrt{d}, \quad\text{for } d\ge2.
\end{align}
\end{corollary}

\begin{proof}
Considering~\eqref{eq:scaling.unvariate.constant} and noting that $|\rho(k)| \le 1$ for all $k \in \za$, we get
\[
	\delta(d) \le \sqrt{d + 2 \sum_{k=1}^{d-1} (d-k)}  = d,		   
\]
where we apply that $\sum_{k=1}^{d-1} (d-k) = \frac12(d^2-d)$.
\end{proof}
The above corollary states that the proportion of the correct scaling factor and the square-root of time when the necessary assumptions are not fulfilled grows with the square-root of time. This conservative estimate does not assume anything non-trivial about the auto-correlations. In Figure~\ref{fig:scalingVMA} in Section~\ref{sec:time.series} we will see that the concrete picture is not always that bad.

\subsection{The multivariate model with auto-correlations}\label{subsec:volacontrib.multi}
Next we will analyze the situation if the constituent assets are known. In this setting a bottom-up
modelling of the portfolio structure is possible.
For the the covariance structure of the asset returns $(\X_t)_{t \in \za}$ we need the following definition:
\begin{definition}[Covariance matrix function of a multivariate time series]
Let $(\X_t)_{t \in \za}$ denote a weakly stationary process in $\re^n$. Then we denote the matrices of serial covariances of lag $k =0,1,2,\ldots$ by
\begin{align}\label{eq:covar.matrix}
	\Gamma(k):=\covar{\X_{t+k}}{\X_t},
\end{align}
for $t \in \za$.
\end{definition}
Consider that the element at position $(i,j)$ in the matrix $\Gamma(k)$ is given by
\[
\Gamma(k)_{i,j} =	\covar{X_{t+k}^i}{X_t^j},
\]
thus modelling multivariate time series it holds that
\[
	\Gamma(k) = \Gamma(-k)^T,
\]
as matrices are not necessarily symmetric. This means that $\covar{X_{t+k}^i}{X_t^j}$ is in general not equal to $\covar{X_{t+k}^j}{X_t^i}$ for $k>0$. In the lead-lag setting the leading market's return `yesterday' is strongly correlated to the lagging one's return `today' but not vice-versa.

Analogously to~\eqref{eq:varvonstrib} the key to volatility contributions in this setting are the covariances of
assets with the portfolio return. The following proposition gives the corresponding expressions for our setting.
%We will directly address the issue of volatility contributions and in order to derive a formula similar to~\eqref{eq:varvonstrib} we first find expressions for the covariance and correlation of lagged returns of asset $i=1,\ldots,n$ with the portfolio.
\begin{prop}\label{prop:multi.var.pf.stationary}
Let $(X_t(\lambda))_{t \in \za}$ denote the portfolio return when weighting the asset returns $(\X_t)_{t \in \za}$ by $\lambda$, i.e.
\[
	X_t(\lambda) = \sum_{i=1}^n \lambda_i X_t^i,\quad\text{for } t \in \za,
\]
then $(X_t(\lambda))_{t \in \za}$ is weakly stationary and
it holds that
\begin{align}\label{eq:multi.covar1}
\gamma_i(k)&:=\covar{X_{t+k}^i}{X_t(\lambda)}= (\Gamma(k) \lambda)_i,
%\rho_i(k)&:=Corr(X_t^i,X_{t+k})= \frac{(\Gamma(k) \lambda)_i}{\sqrt{(\Gamma(0) \lambda)_i}\sqrt{\lambda^T \Gamma(0) \lambda}}, \label{eq:multi.covar2}
\end{align}
for $k =0,1,\ldots$ and $i=1,\ldots,n$ where $(\Gamma(k) \lambda)_i$ denotes the the $i_{th}$ element of the vector $\Gamma(k) \lambda$.
\end{prop}

\begin{proof}
The expectation of $(X_t(\lambda))_{t \in \za}$ is calculated straightforward. Furthermore it holds that
\[
	\covar{X_{t+k}}{X_t(\lambda)} = \covar{\lambda^T \X_{t+k}}{\lambda^T \X_{t}} = \lambda^T \Gamma(k) \lambda,
\]
where $\Gamma(k)$ denotes the covariance matrix of $(\X_t)_{t \in \za}$ for $k=0,1,\ldots$.
The above expression depends on the absolute value of the lag $k$ only as 
\[
	\lambda^T \Gamma(-k) \lambda = \lambda^T \Gamma(k)^T \lambda = \lambda^T \Gamma(k)\lambda,
\]
which concludes the proof of weak stationarity.
To prove~\eqref{eq:multi.covar1} consider that by the bilinearity of covariance we get
\begin{align*}%\label{eq:t1}
\covar{X_{t+k}^i}{X_t(\lambda)} &=\covar{X_{t+k}^i}{\lambda^T \X_t}\\
%&=\sum_{j=1}^n \covar{X_t^i}{\lambda_j X_{t+k}^j}\\ 
&=\sum_{j=1}^n \lambda_j \covar{X_{t+k}^i}{X_t^j} = (\Gamma(k)\lambda)_i,
\end{align*}
for $i = 1,\ldots n$ and $k = 0,1,\ldots$.
%The expressions~\eqref{eq:multi.covar2} for the correlations follows easily.
\end{proof}
Using~\eqref{eq:multi.covar1} we get useful expressions for the auto-covariance structure of the portfolio returns as well as risk contributions in this setting.
\begin{corollary}[Portfolio auto-covariance in the multivariate model]\label{cor:acf.pf.mult}
Let $(X_t(\lambda))_{t \in \za}$ be the portfolio return where the asset returns have covariance structure as given in ~\eqref{eq:covar.matrix} then the auto-covariance function~\eqref{eq:uni.acf} is given by
\begin{align}\label{eq:acf.pf.mult}
	\gamma(k) =  \covar{\lambda^T \X_{t+k}}{\lambda^T \X_t} = \lambda^T \Gamma(k) \lambda= \sum_{i=1}^n \lambda_i \gamma_i(k) 
\end{align}
for $i=1,\ldots,n$ and $k = 0,1,\ldots$ where $\gamma_i(k)$ is given in~\eqref{eq:multi.covar1}.
\end{corollary}
Having calculated the auto-covariance function of $(X_t(\lambda))_{t \in \za}$ we find the volatility scaling for any $d\ge1$ by Proposition~\ref{prop:scaling.univariate} and Corollary~\ref{cor:scaling.univariate}.
To conclude this section we analyze how to calculate contributions to volatility and derive formulas how these contributions change over time.

\begin{prop}[Volatility contributions with serial correlations]\label{prop:vola.contrib.multi}
Let $(\X_t)_{t \in \za}$ denote a weakly stationary process in $\re^n$ and $(X_t(\lambda))_{t \in \za}$ the process of portfolio returns, then the volatility contributions by the Euler-allocation rule when holding this portfolio for $d \ge 2$ days are given by
\begin{align}\label{eq:vola.contrib.multi}
	\sigma_i(\lambda,d) = \frac{\lambda_i}{\sigma(\lambda,d)} \left(d \gamma_i(0) + 2 \sum_{k=1}^{d-1} (d-k) \gamma_i(k)\right),
\end{align}
where $\gamma_i(k)$ is given in~\eqref{eq:multi.covar1} for $i=1,\ldots,n$ and $k = 0,1,\ldots$.
\end{prop}

The following corollary states how risk contributions scale in the above framework.

\begin{corollary}[Scaling of volatility contributions with serial correlations]\label{cor:scaling.vola.corr}
Let $\sigma_i(\lambda,d)$ as in Proposition~\ref{prop:vola.contrib.multi}, then for $i=1,\ldots,n$
\begin{align*}
	\sigma_i(\lambda,d) = \sigma_i(\lambda) \delta(i,d),
\end{align*}
where $\sigma_i(\lambda)$ is the one-day volatility contribution as defined in~\eqref{eq:varvonstrib} and the factor $\delta(i,d)$ is given by
\begin{align}\label{eq:scaling.constant.contrib}
	\delta(i,d) = \left(d + \frac2{\gamma_i(0)} \sum_{k=1}^{d-1} (d-k) \gamma_i(k)\right) / \delta(d),
\end{align}
with $\gamma_i(k)$ defined in~\eqref{eq:multi.covar1} and $\delta(d)$ is the scaling factor for the respective portfolio volatility.
\end{corollary}
The proofs of the two statements above can be found in the appendix.
\begin{remark}
Corollary~\ref{cor:scaling.vola.corr} shows that in this modelling approach the relative volatility contribution changes depending on the holding period and the whole covariance structure of the returns involved. Thus assets whose relative risk contribution increases with the holding period can be identified. This is a feature that neither the model with uncorrelated asset returns nor the univariate model has. Finally, it is easily seen that~\eqref{eq:scaling.constant.contrib} reduces to 
\begin{align*}
		\delta(i,d) = \frac{d}{\sqrt{d}} = \sqrt{d}, \quad\text{for } i = 1,\ldots,n,
\end{align*}
in the case of no auto-correlations.
\end{remark}

Considering the results in Proposition~\ref{prop:scaling.univariate} and~\ref{prop:vola.contrib.multi} in general we have to estimate $d-1$ auto-covariances respectively auto-covariance matrices to find a correct scaling for $d$ days. Considering volatility per annum this corresponds to $d-1=249$ or $251$, which is clearly not appealing.
In the following section we apply these findings to classical auto-regressive time series models which reduces the number of parameters tremendously.
\section{Scaling in ARMA and VARMA Models}
\label{sec:time.series}
First we recall a few well-known definitions from the classical theory of time series analysis, see~\cite{BrowckwellDavis:TimeSeriesandForecasting,McNeil/Frey/Embrechts:RiskManagement,taylor:timeseries}. Note that for the ease of presentation we assume that all asset returns have zero mean. For daily returns this is usually assumed in risk management. We start with the basic building block in time series modelling.

\begin{definition}[White Noise]
A process $(\Z_t)_{t \in \za}$ is called (multivariate) white noise
process $\WN{0}{\Sigma}$ if it is covariance stationary and the covariance matrix function $\Gamma(\cdot)$ is given by
\begin{align*}
	\Gamma(k) = 
	\begin{cases}
		\Sigma, \quad \text{if } k=0, \\
		0, \text{ else,}
	\end{cases}
\end{align*}
where $\Sigma$ is some positive-definite covariance matrix.
\end{definition}
By this definition there are neither serial cross-correlations between component series nor auto-correlations in the case of white noise. The only correlations exist at lag zero. Thus the uncorrelated portfolio case corresponds to a (multivariate) white noise model.

In order to analyze serial correlations we define Box-Jenkins models.
\begin{definition}[Box-Jenkins models]
The weakly stationary process $(\X_t)_{t \in \za}$ with values in $\re^n$ is a called a $\VARMA{p}{q}$ (vector auto-regressive moving average) process if it satisfies the following difference equation 
\begin{align}\label{eq:var.multi}
	\X_t - \sum_{k=1}^p \Phi_k \X_{t-k} = \sum_{j=1}^q \Theta_j \Z_{t-j} + \Z_t, \quad \text{for } t \in \za,
\end{align}
where $(\Z_t)_{t \in \za}$ is $\WN{0}{\Sigma}$ and $\Phi_k,k=1,\ldots,p$ and $\Theta_j,j=1,\ldots,q$ are coefficient matrices in $\re^{n \times n}$. If $\Theta_j = 0$ for $j=1,\ldots,q$, then we denote such a process by $\VAR{p}$, respectively $\VMA{q}$ if $ \Phi_k = 0$  for $k=1,\ldots,p$ . In the one-dimensional case such a process is denoted by $\ARMA{p}{q}$ (auto-regressive moving average) and we write
\begin{align}\label{eq:var.uni}
	X_t - \sum_{k=1}^p \phi_k X_{t-k} = \sum_{j=1}^q \theta_j Z_{t-j} + Z_t, \quad \text{for } t \in \za,
\end{align}
where $(Z_t)_{t \in \za}$ is $\WN{0}{\sigma}$, for some $\sigma>0$.
\end{definition}
In the following we will assume that the processes considered are causal. This means that they have a well defined representation as $\VMA{\infty}$-process, for a definition of causality we refer to the standard textbooks~\cite{TimeSeries:BrockwellDavis,Luethkepohl:NewIntroMultipleTS,mills:timeseries}.

%Interpreting the right-hand side of~\eqref{eq:var.multi} as a matrix-valued polynomial in the lag operator, i.e.
%\begin{align}
%\Phi(z) := I - \Phi_1 z - \cdots - \Phi_p z^p,
%\end{align}
%the notion of causality technically means that
%\[
%	\det(\Phi(z)) \neq 0, \quad\text{for } z \in \cn \text{ with } |z| \le 1. 
%\]
Box-Jenkins models as defined above are frequently used to analyze financial time series and there exist various software packages (e.g.~\cite{CiteR}, Matlab and mathematica) with functions to estimate the parameters as well as to calculate the ACF~\eqref{eq:uni.acf} respectively the ACF-matrices~\eqref{eq:covar.matrix}.

We quote the main propositions that link the parameters in~\eqref{eq:var.multi} to the covariances given in~\eqref{eq:uni.acf} respectively~\eqref{eq:covar.matrix} (see for example~\cite[Chapter 7]{BrowckwellDavis:TimeSeriesandForecasting} or~\cite[Chapter 2]{Luethkepohl:NewIntroMultipleTS}). 
\begin{prop} The cross-covariances of a $\VARMA{p}{q}$-process, denoted by
\[
	\Gamma_{XZ}(k):=\covar{\X_t}{\Z_{t-k}}, \quad \text{for } k \in \za,
\]
%$\Gamma_{XZ}(k):=\covar{\X_t}{\Z_{t-k}}$ with $k \in \za$ 
are given by
\begin{align*}
	\Gamma_{XZ}(k) &= 0, \quad \text{for } k<0,
\end{align*}
and for $k=1,2,3 \ldots$ recursively by
\begin{align}\label{eq:cross.covar}
	\Gamma_{XZ}(k)=\sum_{j=1}^p \Phi_j \Gamma_{XZ}(k-j)  + \Theta_k \Sigma \ind_{k\le q} \quad 	\text{with} \quad \Gamma_{XZ}(0)=\Sigma.
\end{align}
\end{prop}

After this preparation we can state the proposition that finally gives the link needed.
See~\cite{AlbertoMauricio} for a proof and an algorithm for a fast and exact computation of the covariance matrices in the following proposition.

\begin{prop}[Yule-Walker equations]\label{prop:acf}  
%The auto-covariance matrices $\Gamma(k),k=0,1,\ldots$ are uniquely determined by the following system of linear equations:
%
%\begin{align}\label{eq:autovoar.functions}
%\begin{pmatrix} 
%1 & -\Phi_1 & -\Phi_2 & \cdots & -\Phi_{p} \\
%-\Phi_1 & 1-\Phi_2 & -\Phi_3 & \cdots & 0 \\ 
%-\Phi_2 & -\Phi_1-\Phi_3 & 1-\Phi_4 & \cdots & 0 \\
%\vdots &  \ddots & \vdots & \ddots & \vdots \\
%-\Phi_{p} & -\Phi_{p-1} &  -\Phi_{p-2} & \cdots & 1
%\end{pmatrix} \begin{pmatrix} 
%\Gamma(0) \\ \Gamma(1)\\ \Gamma(2)\\ \vdots \\ \Gamma(p) \end{pmatrix} =
%\begin{pmatrix} 
%\sum_{j=0}^q \Theta_j \Gamma_{XY}(j) \\ \sum_{j=1}^q \Theta_j \Gamma_{XZ}(j-1) \\ \sum_{j=2}^q \Theta_j \Gamma_{XZ}(j-2) \\ \vdots \\ \Theta_q \Gamma_{XZ}(0)
%\end{pmatrix} 
%\end{align}
%with $\Theta_0$ the identity matrix and for $k>p$
%\begin{align}\label{eq:auto2}
%\Gamma(k)=\sum_{j=1}^p\Phi_j\Gamma(k-j)+\sum_{j=k}^q \Theta_j \Gamma_{XZ}(j-k).% \ind_{k\le q}.
%\end{align}
%\end{prop}
The auto-covariance matrices $\Gamma(k)$ for $k=0,1,\ldots$ of a causal $\VARMA{p}{q}$ process are uniquely determined by the following system of linear equations
\begin{align}\label{eq:autovoar.functions}
		\Gamma(k) - \sum_{j=1}^p \Phi_j \Gamma(k-j) = \sum_{j=k}^q \Theta_j \Gamma_{XZ}(j-k).
\end{align}
\end{prop}
%For a proof see~\cite{TimeSeries:BoxJenkins} and~\cite{TimeSeries:BrockwellDavis}.

\newpage
In the following we derive concrete formulas for scaling in the univariate case. 
\begin{example}[$\MA{q}$-model]\label{ex:maq}
For the $\MA{q}$-model the formula for ACF is explicitly found in many textbooks (e.g.~\cite[Chapter~4.2]{McNeil/Frey/Embrechts:RiskManagement}) furthermore Proposition~\ref{prop:acf} gives, setting $\theta_0=1$, 
\begin{align}\label{eq:MAk.ACF}
\gamma(0) &=\sigma^2 \sum_{j=0}^q \theta_j^2 , \quad\quad\text{    and, }  \notag \\
\gamma(k) &=\sigma^2 \sum_{j=0}^{q-|k|} \theta_j \theta_{j+|k|},  \quad\text{for } |k|=1,2, \ldots,q.
\end{align}
Obviously $\gamma(k) = 0$ for $|k| > q$.

We can plug the acf~\eqref{eq:MAk.ACF} into the formula~\eqref{eq:scaling.unvariate.constant} and see that the correct scaling factor in this setting is given by
\begin{align*}
	\delta(d) = \sqrt{d + 2  \sum_{k=1}^{Min(d-1,q)} (d-k) \frac{\sum_{j=0}^{q-|k|} \theta_j \theta_{j+|k|}}{\sum_{j=0}^q \theta_j^2}  }.
\end{align*}
\end{example}

\begin{example}[$\AR{1}$-model]\label{example:ARMA1}
Incorporating an auto-regressive coefficient we consider an $\AR{1}$-model whose ACF~\footnote{With increasing order of the models handy formulas for the auto-covariance function can not be provided any more, but the function \texttt{ARMAacf} in the software package R (see~\cite{CiteR}) or similar implementations may be used to find the ACF after having estimated the ARMA-coefficients.} is given by~\cite[Chapter~4.2]{McNeil/Frey/Embrechts:RiskManagement}
\begin{align*}
	\gamma(k)  = \frac{\phi_1^{|k|} \sigma^2}{1-\phi_1^2},  \quad\text{for } |k| =0, 1,2, \ldots
\end{align*}
Note that in this case $\gamma(k) \neq 0$ for $k \in \za$ and we have to apply standard summation calculus to get 
the following scaling factor
\begin{align*}
	\delta(d) &= \sqrt{d + 2 \frac{\phi_1}{(\phi_1-1)^2} (d(1-\phi_1) + \phi_1^d-1) }.
\end{align*}

\end{example}

\subsection{Closing time problem}

Before we consider the closing time problem it is essential to understand the interplay between multivariate time series and portfolios constructed from these. Creating a portfolio return  $(X_t)_{t\in\za}$ from asset returns modelled by a multivariate time series $(\X_t)_{t\in\za}$ using the weights $\lambda$ by $X_t = \lambda^T \X_t$  means to apply a linear transformation to $\X_t$. To understand this transformation we quote the following proposition on $\VMA{q}-$processes~\cite[Proposition~ 11.1]{Luethkepohl:NewIntroMultipleTS}:

\begin{prop}[Linear transformation of a $\VMA{q}$-process]
Let $(\Z_t)_{t \in \za}$  be an n-dimensional white noise process with nonsingular covariance
matrix $\Sigma$ and let
\[
	\X_t = \sum_{j=1}^q \Theta_j \Z_{t-j} + \Z_t, \quad \text{for } t \in \za,
\]
be an n-dimensional invertible $\VMA{q}$-process. Furthermore, let F be an $(M \times
n)$ matrix of rank M. Then the M-dimensional process $\Y_t = F \X_t$ has an
invertible $\VMA{\tilde{q}}$ representation,
\[
\Y_t = \sum_{j=1}^{\tilde{q}} \tilde{\Theta}_j  \Ztilde_{t-j} + \Ztilde_t, \quad \text{for } t \in \za,
\]
where $(\Ztilde_t)_{t \in \za}$ is M-dimensional white noise with nonsingular covariance matrix, the $\tilde{\Theta}_j$ are coefficient matrices and $\tilde{q} \le q$.
\end{prop}
This proposition can be applied to analyze the situation if we form a portfolio. Then $F = \lambda^T$  is $1 \times n$ and we get a moving average process of order equal or less the order of the vector process. For the closing time problem $q = 1$ which leads to $\tilde{q} = 1$ and we can summarize our findings for this problem so far:
\begin{obs}\label{obs:obsMA}
The multivariate process of closing-time returns of assets traded in different time zones can be modelled as a $\VMA{1}$ process of the form
\[
	\X_{t}  = \Theta_1 \Z_{t-1} + \Z_{t} , \quad \text{for } t \in \za.
\]
Creating a portfolio with asset weights $\lambda$ results in a $\MA{1}$ process of the from
\[
\lambda^T \X_{t} =:  X_{t}  = \theta_1 Z_{t-1} + Z_{t} , \quad \text{for } t \in \za.
\]
\end{obs}
For a $\MA{1}$-process and the scaling constant of Example~\ref{ex:maq} simplifies to
\begin{align}\label{eq:dMA1}
	\delta(d) = \sqrt{d + 2 (d-1) \frac{\theta_1}{1+\theta_1^2} }, \quad \text{for } d > 1.
\end{align} 

This allows a top-down modelling of the returns if we think that the only auto-correlations in the portfolio come from the closing time problem. We just have to fit the parameters $\theta$ and $\sigma^2$ by~\eqref{eq:MAk.ACF} and apply~\eqref{eq:dMA1} for correct scaling of portfolio volatility. Alternatively, we can of course directly estimate the autocorrelation $\rho(1)$ of portfolio returns and plug it into~\eqref{eq:scaling.unvariate.constant}. However, estimating the full model would give as more insight.

\paragraph{The $\VMA{1}$-model}
In this paragraph we derive the details of the formulas of Section~\ref{sec:volacontrib} for the $\VMA{1}$-case. 
Below in Example~\ref{ex:portfolio2} we apply these formulas to the closing time problem started in  Example~\ref{ex:portfolio}.
By~\eqref{eq:cross.covar} and~\eqref{eq:autovoar.functions} we get the following for the $\VMA{1}$-model:
\begin{align}
	\Gamma(0) &= \Theta_1 \Sigma \Theta_1^T +  \Sigma, \quad \text{ and} \notag \\
	\Gamma(1) &= \Theta_1 \Sigma, \label{eq:GammaVMA} 
\end{align}
which gives 
\begin{align}
	\gamma(0) &= \lambda^T \left(\Theta_1 \Sigma \Theta_1^T +  \Sigma\right) \lambda, \quad \text{ and}\notag  \\
	\gamma(1) &=\lambda^T \left( \Theta_1 \Sigma \right) \lambda,  \label{eq:gammaVMA}
\end{align}
for the portfolio time series by~\eqref{eq:acf.pf.mult}.
Thus the scaling constant~\eqref{eq:scaling.unvariate.constant} is given by
\begin{align}  \label{eq:VMA.deltad.pf}
	\delta(d) = \sqrt{d + 2 (d-1) \frac{\lambda^T  \Theta_1 \Sigma \lambda}{ \lambda^T \left(\Theta_1 \Sigma \Theta_1^T +  \Sigma\right) \lambda} },
\end{align}
and 
the scaling of contributions~\eqref{eq:scaling.constant.contrib} is given by
\begin{align}  \label{eq:VMA.deltad.contrib}
	\delta(i,d) = \left(d + 2 (d-1) \frac{ (\Theta_1 \Sigma \lambda)_i}{\left((\Theta_1 \Sigma \Theta_1^T +  \Sigma) \lambda\right)_i} \right) / \delta(d),
\end{align}
where $\delta(d)$ is calculated in~\eqref{eq:VMA.deltad.pf} above. Again as in Example~\ref{ex:maq} one can estimate $\Gamma(0)$ and $\Gamma(1)$ and plug them into~\eqref{eq:scaling.constant.contrib}. But the estimator for $\Theta_1$ and $\Sigma$ or transformations of it will reveal interesting structures (see~e.g.~\cite[Chapter~8.2.1 Reduced and Structural Forms]{tsay2005analysis} or~\cite[Chapter~2.3.2 Impulse Response Analysis]{Luethkepohl:NewIntroMultipleTS}).

\paragraph{Scaling volatility of closing-time returns in $\VMA{1}$ compared to scaling volatility of contemporaneous returns}
In this paragraph we have a closer look at the Newey-West estimator~\eqref{eq:NWintro_real} and the na\"{\i}ve estimator~\eqref{eq:NWintro} in the the context of volatility scaling.
Having estimated the lag-zero covariance matrix of asset-returns $\Gamma(0)$ and the lag-one covariance matrix $\Gamma(1)$ we find the volatility of the $d$ days portfolio closing-time return by~\eqref{eq:acf.pf.mult} as 
\begin{align}\label{eq:comp.eq1}
	\sigma(\lambda,d) =  \sqrt{d \lambda^T \Gamma(0) \lambda + 2 (d-1) \lambda^T \Gamma(1) \lambda}, \quad \text{for } d > 1.
\end{align}

Considering the contemporaneous returns together with their covariance matrix given by the na\"{\i}ve estimator $\tilde{\Sigma}$ from~\eqref{eq:NWintro} and assuming zero autocorrelations among contemporaneous returns the 
$d$ days volatility of the contemporaneous portfolio return $\tilde{\sigma}(\lambda,d)$ is given by  
\begin{align}\label{eq:comp.eq2}
 \tilde{\sigma}(\lambda,d)	&= \sqrt{d  \lambda^T \tilde{\Sigma} \lambda} \notag \\
          &= \sqrt{d  \lambda^T (\Gamma(0) + \Gamma(1) + \Gamma(1)^T) \lambda}   \notag \\
	&= \sqrt{d \lambda^T \Gamma(0) \lambda + 2 d \lambda^T \Gamma(1) \lambda}, \quad \text{for } d > 1.
\end{align}
Considering the ratio of the scaled volatility of the portfolio closing-time return~\eqref{eq:comp.eq1} over the scaled volatility of the contemporaneous portfolio return~\eqref{eq:comp.eq2} we see that this quantity converges to 1:
\begin{align}\label{eq.naivelimit1}
	\lim_{d \rightarrow \infty} \frac{\sigma(\lambda,d)}{\tilde{\sigma}(\lambda,d)}  = \lim_{d \rightarrow \infty} \frac{\sqrt{d \lambda^T \Gamma(0) \lambda + 2 (d-1) \lambda^T \Gamma(1) \lambda}}{\sqrt{d \lambda^T \Gamma(0) \lambda + 2 d \lambda^T \Gamma(1) \lambda}} = 1.
\end{align}
Thus for large $d$ the risk figures for the two procedures coincide.

However, applying the Newey-West estimator up to lag~1~\eqref{eq:NWintro_real} the limit of the corresponding ratio is given by
\begin{align*}
 \lim_{d \rightarrow \infty}&\frac{\sqrt{d \lambda^T \Gamma(0) \lambda + 2 (d-1) \lambda^T \Gamma(1) \lambda}}{\sqrt{d \lambda^T \Gamma(0) \lambda + d \lambda^T \Gamma(1) \lambda}}    \\
 =& \frac{\sqrt{ \lambda^T \Gamma(0) \lambda + 2 \lambda^T \Gamma(1) \lambda}}{\sqrt{\lambda^T \Gamma(0) \lambda +  \lambda^T \Gamma(1) \lambda}} \neq 1.
\end{align*}
Thus the volatility estimates by the Newey-West estimator~\eqref{eq:NWintro_real} for contemporaneous returns, in this set-up, do not coincide with the result of applying the $\VMA{1}$-model for closing-time returns nor with the result of applying the na\"{\i}ve (but in this case correct and useful) estimator~\eqref{eq:NWintro}. This shows that the compatibility of long-term risk estimates in the two notions, closing-time return and contemporaneous return, depends on the estimator used for the covariance matrix of contemporaneous returns.

As already mentioned the estimator $\tilde{\Sigma}$ defined in~\eqref{eq:NWintro} can, in general, be invalid which by~\eqref{eq.naivelimit1} questions the $\VMA{1}$-model. In this case we should analyze the situation in more detail and make sure that the assumption that the closing-time problem is the only source of auto-correlation is acceptable. In the following we apply our findings to the data from Example~\ref{ex:portfolio} and all estimators considered are mathematically valid.

%
%\begin{example}[Example~\ref{ex:portfolio} continued]\label{ex:portfolio2}
%We come back to our problem of aronous closing time. Motivated by the vanishing covariance matrices of lag two and greater we assume an $\MA{1}$-model. For a portfolio with weights $\lambda$ we thus have by from~\eqref{eq:MAAssets}
%\begin{align}
%	\lambda^T \X_{t+1} &= \lambda^T \Theta \Z_t + \lambda^T \Z_{t+1} , \quad \text{for } t \in \za \\
%           X_{t+1} &= \theta Z_t +  Z_{t+1} , \quad \text{for } t \in \za,
%\end{align}
%where $(Z_t)_{t \in \za}$ is a univariate white noise process with variance $\lambda^T \Sigma \lambda$, the MA-coefficient
%is given by 
%\[
%	\theta = \frac{\lambda^T \Theta \Sigma \Theta^T \lambda}{\lambda^T \Sigma \lambda},
%\] 
%and $\Sigma$ is the covariance matrix of the multivariate white noise process $(\Z_t)_{t \in \za}$.
%\end{example}

\begin{example}[Example~\ref{ex:portfolio} continued]\label{ex:portfolio2}
The conclusion of Example~\ref{ex:portfolio} is that we can model the closing-time returns as $\VMA{1}$-process.
We apply this to the data of Example~\ref{ex:portfolio}. Note that we can not solve~\eqref{eq:GammaVMA} for $\Theta_1$ or $\Sigma$. Thus we apply  maximum likelihood estimation for this task as it is provided in the R-package DSE~\cite{Gilbertdse}. Note that positive definiteness of the covariance matrix of residuals of the $\VMA{1}$-model is assured in the estimation procedure applied \footnote{We acknowledge personal communications with the author, Paul Gilbert, on this topic.}. This fact is required for valid estimators in~\eqref{eq:GammaVMA}.

On the other hand we can directly estimate $\Gamma(0)$ and $\Gamma(1)$ from the data, knowing that $\Gamma(k) = 0$ for $k\ge2$. But looking at $\Theta_1$ and $\Sigma$  gives us some insight in the problem at hand. As risk and risk contributions is the focus of this article we do not go through the whole impulse-response analysis but refer to~\cite[Chapter~8.2.1 Reduced and Structural Forms]{tsay2005analysis} or~\cite[Chapter~2.3.2 Impulse Response Analysis]{Luethkepohl:NewIntroMultipleTS}) or the generalized impulse-response of~\cite{Pesaran:genirf}.

In Figure~\ref{fig:scalingVMA} we see the scaling constant for the portfolio time-series by the just mentioned MLE estimation of the $\VMA{1}$-model of the seven assets, a $\MA{1}$-model estimated directly on the portfolio time-series and the square-root-rule. The reason why the upper lines do not match perfectly are estimation errors but this plot gives us some confidence for the MLE estimators for $\Theta_1$ and $\Sigma$.

\begin{figure}[htbp]
	\begin{center}
	\includegraphics[width=0.8\textwidth]{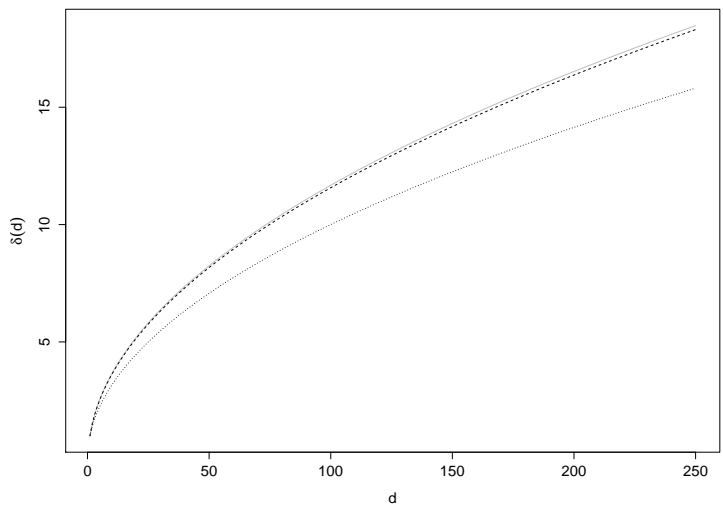}
	\end{center}
	\caption{Scaling constants $\delta(d)$ for ranging $d$ for a full $\VMA{1}$-model (solid, gray), a univariate $\MA{1}$-model (dashed, black) and the square-root-of-time (dotted, black).}
	\label{fig:scalingVMA}
\end{figure}

Considering Table~\ref{tab:pf.contrib} we see the contributions to volatility p.a. on one hand by assuming zero serial correlations between the assets in the portfolio and on the other hand by modelling the asset returns in the portfolio as $\VMA{1}$-process.
Applying the square-root-of-time rule to this portfolio we get a volatility p.a. of $ 20.07\%$ while the volatility p.a. increases by approximately $17\%$ to $ 23.46\%$ in the $\VMA{1}$-model. Concerning the analysis of sources of risk note that the risk contribution by the geographically most distant market Japan of $1.58\%$ looks quite small when ignoring serial cross-correlations but it increases to $3.84\%$ taking them into account. The risk contributions of the markets leading the portfolio do not change dramatically. The increase of the portfolio volatility can be attributed to the Asian assets whose risk contribution increases significantly as serial cross correlations to the US and Europe are taken into account. This example shows that not only the accuracy of total volatility of the portfolio increases but also the attribution to single assets fits economic considerations much better!

\begin{table}[!htb]
	\begin{center}
		\begin{tabular}{l |c c c c c}
			 Asset & Currency & Exposure & Square-root rule &  $\VMA{1}$ & Difference \\ 
			 \hline\hline
			  Portfolio    &  EUR   & 105\%       &    20.07\% & 23.46\%  & 3.38\% \\ 
			  \hline
        Topix &  JPY   &15\%      &    1.58\% &  3.84\%  & 2.26\% \\
        H-shares &  HKD &  15\%   &    3.56\% &  5.40\%  & 1.84\% \\                  
DJ Euro Stoxx 50 & EUR  &  15\%   &    3.28\% &  2.96\%	& -0.32\% \\
      Swiss Market &  CHF   & 15\%&    2.24\% &  2.11\%	& -0.13\% \\ 
   JSE TOP 40 &  ZAR   & 15\%   &      2.62\% &  2.96\%	&  0.34\% \\  
  Russell 2000  &  USD    &15\%    &   3.94\% &  3.43\%	& -0.51\% \\
      NASDAQ 100 &  USD   &15\%  &     2.85\% &  2.75\%	& -0.11\%
		\end{tabular}
	\end{center}
	\caption{Contributions to volatility p.a. in a the global portfolio of Example~\ref{ex:portfolio} applying the square-root-of-time rule and modelling a $\VMA{1}$-process.}
		\label{tab:pf.contrib}
\end{table}

\end{example}

\subsection{Genuine auto-correlations}
We conclude our theoretical study of auto-correlated portfolio returns by a short detour to genuine auto-correlations. In the literature studies can be found (see e.g~\cite{Andersonetal:StockReturn} and references therein) which provide evidence that, besides the spurious effects of non-contemporaneous trading, genuine effects such as partial price adjustment and time-varying risk premia can lead to genuine auto-correlations in asset returns. We stay in the class of vector-autoregressive models but focus on $\VAR{p}$, especially, $\VAR{1}$-models in this context as opposed to the $\VMA{1}$-model for the closing time problem above.

Before we consider an example of first order genuine auto-correlations we go one step deeper into understanding the interplay between multivariate time series and portfolios constructed from these. An important result on linear transformations of $\VARMA{p}{q}$ processes is the following~\cite[Corollary~ 6.1.1]{Luethkepohl:NewIntroMultipleTS}:
\begin{theorem}[Linear transformations of $\VARMA{p}{q}$ processes]\label{th:varmaqptilde}
Let $(\X_t)_{t\in\za}$ be an $n$-dimensional, stable, invertible $\VARMA{p}{q}$ process and let $F$ be an $M \times n$ matrix of rank $M$. Then the process $(F \X_t)_{t\in\za}$ has a $\VARMA{\tilde{p}}{\tilde{q}}$ representation with
\begin{align*}
	\tilde{p} \le n p \quad \text{and} \quad \tilde{q} \le (n-1)p + q.
\end{align*}
\end{theorem}
The above theorem tells us that a portfolio which is a simple linear transformation of a $\VARMA{p}{q}$ process can not be guaranteed to have an $\ARMA{p}{q}$ representation of the same order. Furthermore it is important to note that as a consequence the class of $\VAR{p}$-models is not closed with respect to linear transformation as, in general, the result of the transformation can be some $\VARMA{\tilde{p}}{\tilde{q}}$ process with $\tilde{q}>0$.
Concerning multivariate models we nevertheless focus on $\VAR{p}$-models due to known identification problems of 
$\VARMA{p}{q}$-models if $q>0$ (see for example~\cite{Luetkepohl:Forecast}). 

As a preparation for our key result on portfolios constructed out of $\VAR{p}$-process we state the following proposition:

%\begin{prop}[AR portfolios built from VAR processes]\label{prop:VAR.AR.1}
%Let $(\X_t)_{t \in \za}$ be a $\VAR{p}$-process in $\re^n$ for $n>0$ of the form
%\[
%	\X_t = \sum_{k=1}^p \Phi_k \X_{t-k} + \Z_t
%\]
%and let $\lambda = (\lambda_1,\ldots,\lambda_n)$ be a vector of weights in $\re^n$. Then the portfolio process $(X_t(\lambda))_{t \in \za}$ is an $\AR{p}$-process of the form
%\[
%	X_t(\lambda) = \sum_{k=1}^p \phi_k X_{t-k}(\lambda) + Z_t
%\]
%iff
%\begin{enumerate}
%\item $Z_t = \lambda^T \Z_t$ for $t \in \za$ and \label{prop:VAR.AR.1.item1}
%\item $\lambda^T (\Phi_k - \phi I_n) = 0$ for $k=1,\ldots,p$.   \label{prop:VAR.AR.1.item2}
%\end{enumerate}
%\end{prop}
\begin{prop}[AR portfolios built from VAR processes]\label{prop:VAR.AR.1}
Let $(\X_t)_{t \in \za}$ be a $\VAR{p}$-process in $\re^n$ for $n>0$ of the form
\[
	\X_t = \sum_{k=1}^p \Phi_k \X_{t-k} + \Z_t
\]
and let $\lambda = (\lambda_1,\ldots,\lambda_n)$ be a vector of weights in $\re^n$. Then the portfolio process $(X_t(\lambda))_{t \in \za}$ is an $\AR{p}$-process of the form
\begin{align}\label{eq:VAR.AR.1.ARP}
	X_t(\lambda) = \sum_{k=1}^p \phi_k X_{t-k}(\lambda) + \lambda^T \Z_t
\end{align}
if and only if
\begin{align}\label{prop:VAR.AR.1.item2}
	\lambda^T \Phi_k  =  \phi_k \lambda^T\quad \text{for } k=1,\ldots,p.   
\end{align}
\end{prop}
%The proof of this proposition can be found in the appendix. While the first condition in the above proposition is mainly a technical one the second means that only such portfolios preserve the order of the (V)AR process that have the property to be some kind of eigenvector to all coefficient matrices. The following corollary concludes these considerations.
Note that as in Corollary~\ref{cor:acf.pf.mult} $(\lambda^T \Z_t)_{t \in \za}$ is clearly a white noise process. The full
proof of the proposition can be found in the appendix. Condition~\eqref{prop:VAR.AR.1.item2} means
that only portfolios with $\lambda$ being an eigenvector of all coefficient
matrices of the $\VAR{p}$-process admit the $\AR{p}$ representation~\eqref{eq:VAR.AR.1.ARP} which one could expect to hold in general, at first glance. The following corollary concludes these considerations.
%The proof of this proposition can be found in the appendix. While the first condition in the above proposition is mainly a technical one the second means that only such portfolios preserve the order of the (V)AR process that have the property to be some kind of eigenvector to all coefficient matrices. 
\begin{corollary}[AR portfolios built from VAR processes]\label{cor:VAR.AR.1}
In the setting of Proposition~\ref{prop:VAR.AR.1} the process $(X_t(\lambda))_{t \in \za}$ is an $\AR{p}$-process for any portfolio weighting $\lambda \in \re^n$ if and only if the coefficient matrices of the VAR-process are diagonal and of the following form
\[
	\Phi_k = \phi_k I_n, \quad\text{for } k=1,\ldots,p.
\]
\end{corollary}
The consequence of the above corollary is that modelling genuinely auto-correlated assets we will in general not observe portfolio returns consistent with an $\AR{1}$-model. This would only be possible if the coefficient matrix $\Phi_1$ were of the form
\[
	\begin{pmatrix} a &    0  & \dots      &    0       \\
	                0 &    a  & \dots      &    0       \\
							\vdots&    0   & \ddots     & \vdots \\
							    0 & \dots & \dots      & a 			\\
   \end{pmatrix}
\]
for some fixed value for $a$ - all the same for each asset.
%Considering Example~\ref{ex:portfolio} it would be inconsistent to assume that dependence on lagged values is univariate, which is implied by a %diagonal structure of $\Phi_1$. Furthermore the assumption that this dependence is the same for all assets is clearly unrealistic.
The conclusion is that the weighted $\VAR{1}$-model is richer than an $\AR{1}$-model.

\paragraph{The $\VAR{1}$-model}
After these general considerations we focus on the $\VAR{1}$-model of the form
\[
	\X_t = \Phi_1 \X_{t-1} + \Z_t, \quad \text{for } t\in \za,
\]
since this model will be the natural choice to capture genuine serial correlations.
In this case we get the following expressions for the covariance matrices
\begin{align}
	\Gamma(0) &= \Sigma (I - \Phi_1^2)^{-1} \quad\text{and} \nonumber \\  
	\Gamma(k) &= \Phi_1^k \Gamma(0), \quad \text{for } k\ge1,  \label{eq.var1.2}
\end{align}
which gives 
\begin{align}
	\gamma(0) &= \lambda^T \Gamma(0) \lambda, \quad \text{ and}\notag  \\
	\gamma(k) &=\lambda^T \left( \Phi_1^k \Gamma(0) \right) \lambda, \quad \text{for }  k\ge1,   \label{eq:gammaVAR}
\end{align}
for the portfolio time series by~\eqref{eq:acf.pf.mult}.
Using~\eqref{eq.var1.2} and~\eqref{eq:gammaVAR} the scaling constant~\eqref{eq:scaling.unvariate.constant} for $d>1$ is given by
\begin{align}  \label{eq:VAR.deltad.pf}
	\delta(d) = \sqrt{d + 2 \sum_{k=1}^{d-1} (d-k) \frac{\lambda^T \Phi_1^k \Gamma(0) \lambda}{\lambda^T \Gamma(0) \lambda}},
\end{align}
and the scaling of contributions~\eqref{eq:scaling.constant.contrib} is given by
\begin{align}  \label{eq:VAR.deltad.contrib}
	\delta(i,d) = \frac1{\delta(d)} \left(d + 2 \sum_{k=1}^{d-1} (d-k) \frac{( \Phi_1^k \Gamma(0) \lambda)_i}{(\Gamma(0) \lambda)_i} \right),
\end{align}
where $\delta(d)$ is calculated in~\eqref{eq:VAR.deltad.pf} above. In contrast to the closing-time problem, in this case, we would have to estimate $d-1$ lagged covariance matrices $\Gamma(k), k = 1,\ldots, d-1$ if we wanted to plug them into~\eqref{eq:scaling.constant.contrib} directly. This is clearly not feasible for large values of $d$ which justifies the use of a specific time-series model in these cases.

The following concrete example illustrates the above issues. We furthermore analyse
the trade-off when approximating such a portfolio with first order genuine auto-correlations with an $\AR{1}$-model, although it is not theoretically justified.

\begin{example}\label{example:pfAB}
Consider a portfolio consisting of two contemporaneously traded assets A and B with annual volatilities of $25\%$ and $20\%$ and a correlation of $70\%$. Asset A has a negative genuine first order autocorrelation of $-5\%$ while asset B exhibits $2.5\%$ genuine first order autocorrelation. Note that these values are consistent with findings in~\cite{Andersonetal:StockReturn}. Thus we consider the following covariance matrices
\begin{align*}
 \Gamma(0) &=\text{diag}\begin{pmatrix} 
0.25 & 0.2 \\
\end{pmatrix} \begin{pmatrix} 
1 & 0.7 \\
0.7 &1 \\
\end{pmatrix}
\text{diag}\begin{pmatrix} 
0.25 & 0.2 \\
\end{pmatrix} \frac1{250} \quad \text{ and } \\
   \Gamma(1) &= 
\text{diag}\begin{pmatrix} 
0.25 & 0.2 \\
\end{pmatrix} \begin{pmatrix} 
-0.05 & 0 \\
0 & 0.025 \\
\end{pmatrix}
\text{diag}\begin{pmatrix} 
0.25 & 0.2 \\
\end{pmatrix} \frac1{250}.
\end{align*}
We can model these two assets by a $\VAR{1}$-model and calculate the coefficient matrix $\Phi_1$ by~\eqref{eq.var1.2} and get
\[
	\Phi_1 = \Gamma(1) \Gamma(0)^{-1} = \begin{pmatrix} 
-0.0980 & 0.0858 \\
-0.0275 & 0.0490 \\
\end{pmatrix}.
\]
Modelling a portfolio with the weighting $\lambda = (\frac12,\frac12)^T$ we expect an ARMA(2,1) process by Theorem~\ref{th:varmaqptilde}.  However, for a pure $\AR{1}$-model the coefficient $\phi_1$ is given by
\[
 \phi_1 = \gamma(1)/\gamma(0) = \frac{2.362}{14.53}  =  -0.0123.
\]
Using the results for scaling in univariate models from Example~\ref{example:ARMA1} and Equation~\eqref{eq:VAR.deltad.pf} we compare the resulting scaling constants in Table~\ref{table:scaling.comparison}. We see that the $\AR{1}$-model performs well especially for shorter holding periods in approximating the result of the multivariate model. However, such a simple approximation should be used with care. Furthermore the $\VAR{1}$-model tells us more details about the risk contributions as we will see below.
\begin{table}[!htb]
	\centering
		\begin{tabular}{l|c c c}
			d &  $\VAR{1}$ & $\AR{1}$ & SRTR \\ \hline
2 &  1.405  &1.405 &1.414  \\
5 & 2.218  &2.214 & 2.236  \\
10 &  3.134 & 3.127 &3.162  \\
30 & 5.427  &5.412 & 5.477  \\ 
90 &  9.398 & 9.372 & 9.487 \\
20 & 15.662 &15.619 &15.811		\end{tabular}    
		\caption{Volatility scaling factors $\delta(d)$ for $\VAR{1}$, $\AR{1}$ and the square-root of time for various holding periods $d$.}
		\label{table:scaling.comparison}
\end{table}

We conclude this detour on genuine auto-correlations by an analysis of the relative risk contributions, i.e. risk contributions in percentage of total volatility. Applying~\eqref{eq:VAR.deltad.contrib} to our example we see in Table~\ref{table:scaling.risk.contributions} that the contribution of asset A is dominant as it has the higher volatility. But with increasing holding period the relative risk contribution of asset A decreases which reflects its negative auto-correlation and the positive auto-correlation of asset $B$. This is a feature that only the multivariate approach can offer.

\begin{table}[!htb]
	\centering
		\begin{tabular}{l|c c }
			$d$ & $ \frac{\sigma(d,A)}{\sigma(d)} $ & $ \frac{\sigma(d,B)}{\sigma(d)} $ \\ \hline
1  & 56.52 &43.48\\
2 & 55.39 &44.61\\
5 &  54.77 &45.23\\
10 & 54.56 &45.44\\
30 &54.42 &45.58\\
90 & 54.37 &45.63\\
250 & 54.36& 45.64\\
		\end{tabular}   
		\caption{Relative risk contributions (percentage) of assets A and B for various holding periods $d$ in the $\VAR{1}$-model.}
		\label{table:scaling.risk.contributions}
\end{table}

\end{example}

%After this in-depth analysis of the $\VAR{1}$-model we conclude and come back to our real world global portfolio of %Example~\ref{ex:portfolio}. 

\section{Conclusions}
In this article we first clarify the notion of closing-time returns and contemporaneous returns in global portfolios. In Example~\ref{ex:portfolio} we illustrate these notions in a setting of time-shifted multivariate Brownian motion. Serial correlations naturally occur when analyzing portfolios of geographically diversified assets traded in distant time zones and we motivate the use of a $\VMA{1}$-model for the closing-time returns. We then address the problem of calculating portfolio volatility of closing-time returns for holding periods of more than one day.

We show that ignoring serial correlations leads on one hand to biased estimates of volatility and on the other hand
to misleading risk contributions as Example~\ref{ex:portfolio2} illustrates.  %In Example~\ref{ex:portfolio} we motivate the %use of $\VAR{1}$-models in such situations.

We propose formulas for calculating accurate volatility scaling modelling the portfolio closing-time return as a univariate process as well as in a multivariate setting. Moreover in the multivariate setting we also provide explicit formulas for genuine risk contributions that take the time series structure of the assets involved into account. Modelling the asset returns as a vector moving average process of order one we derive handy formulas and perform a complete analysis of risk and risk contributions and compare this approach to the Newey-West estimator of contemporaneous returns and another simple but useful estimator in the same spirit.

Finally we take a short detour to genuine auto-correlations and propose the application of a $\VAR{1}$-model to tackle this problem.

Applying the findings of this article to the calculation of the tracking error, i.e. the volatility of the additional return of the portfolio above a given benchmark, can improve the analysis of relative risk which is often an aim in asset management.

Besides the analysis of market risk our findings can be applied to portfolio optimization as well as portfolio construction techniques such as risk-parity (also known as equally-weighted risk contributions, see~\cite{Roncalli:PropEquallyWeighted}) where risk contributions by assets are the driving input.  
As another direction of further research the findings of this article may also be applied to the VEC specification of 
multivariate GARCH models, since they admit a VARMA representation (see~\cite{Luethkepohl:NewIntroMultipleTS}).

\appendix
\section*{Appendix: Proofs}
\begin{proof}[Proof of Proposition~\ref{prop:scaling.univariate} and Corollary~\ref{cor:scaling.univariate}]
Considering that $\sigma(\lambda,d)$ is the square-root of $\var{\sum_{i=1}^d X_i(\lambda)}$ and writing down this variance in matrix form using 
weak stationarity we get
\begin{align*}
\var{\sum_{i=1}^d X_i(\lambda)} &=\textbf{1}^T \begin{pmatrix} \covar{X_1(\lambda)}{X_1(\lambda)} & \dots & \covar{X_1(\lambda)}{X_d(\lambda)}\\
\covar{X_2(\lambda)}{X_1(\lambda)} &\dots & \covar{X_2(\lambda)}{X_d(\lambda)}\\
\vdots&\ddots&\vdots\\
\covar{X_d(\lambda)}{X_1(\lambda)} &\dots & \covar{X_d(\lambda)}{X_d(\lambda)}\end{pmatrix} \textbf{1} \\
&=\textbf{1}^T \begin{pmatrix} \gamma(0) & \gamma(1) &\dots & \gamma(d-1)\\
\gamma(1) & \gamma(0) &\dots & \gamma(d-2)\\
\vdots&\vdots&\ddots&\vdots\\
\gamma(d-1) & \gamma(d-2) &\dots & \gamma(0)\end{pmatrix} \textbf{1},
\end{align*} 
where $\gamma(\cdot)$ denotes the auto-covariance function of $(X_t(\lambda))_{t \in \za}$ and  $\textbf{1}=(1,\dots,1)^T$. Summing up along the diagonals and using symmetry we get Equation~\eqref{eq:scaling.unvariate}.
For proving~\eqref{eq:scaling.unvariate.constant} note that the auto-covariances can be expressed in terms of the auto-correlation and the variance in the following sense:
\begin{align*}
	\rho(k)=\frac{\gamma(k)}{\gamma(0)} \Leftrightarrow \gamma(k) = \gamma(0) \rho(k) = \sigma(\lambda)^2 \rho(k),
\end{align*}
for $k = 0,1,\ldots$
\end{proof}

\begin{proof}[Proof of Proposition~\ref{prop:vola.contrib.multi} and Corollary~\ref{cor:scaling.vola.corr}]
Following the Euler allocation rule the volatility contributions are given by
\[
	\sigma_i(\lambda,d) = \lambda_i \frac{\covar{\sum_{k=1}^d X^i_k}{\sum_{k=1}^d X_k(\lambda)}}{\sigma(\lambda,d)},
\]
where we can calculate the denominator $\sigma(\lambda,d)$ by~\eqref{eq:scaling.unvariate} and~\eqref{eq:acf.pf.mult}.
For the numerator we get
\begin{align*}
\covar{\sum_{k=1}^d X^i_k}{\sum_{k=1}^d X_k(\lambda)} &= \textbf{1}^T \begin{pmatrix} \covar{X_1^i}{X_1(\lambda)} &\dots & \covar{X_1^i}{X_{d}(\lambda)}\\
\covar{X_2^i}{X_1(\lambda)} &\dots & \covar{X_2^i}{X_{d}(\lambda)} \\
\vdots&\ddots&\vdots\\
\covar{X_d^i}{X_1(\lambda)} &\dots & \covar{X_d^i}{X_d(\lambda)}\end{pmatrix} \textbf{1}\\
           &=  \textbf{1}^T \begin{pmatrix} \gamma_i(0) & \gamma_i(1) &\dots & \gamma_i(d-1)\\
\gamma_i(1) & \gamma_i(0) &\dots & \gamma_i(d-2)\\
\vdots&\vdots&\ddots&\vdots\\
\gamma_i(d-1) & \gamma_i(d-2) &\dots & \gamma_i(0)\end{pmatrix} \textbf{1},        
\end{align*}
where $\textbf{1} = (1,\ldots,1)^T$ and the $\gamma_i(k)$ are given in~\eqref{eq:multi.covar1} for $i=1,\ldots,n$ and $k = 0,1,\ldots$
Now again use the symmetries and sum up along the diagonals to get the result.
To prove the corollary, recall that $\sigma(\lambda,d) = \sigma(\lambda) \delta(d)$. Plugging this
into the denominator of~\eqref{eq:vola.contrib.multi} and recalling that the one day risk contribution is given by $\sigma_i(\lambda) = \lambda_i \frac{\gamma_i(0)}{\sigma(\lambda)}$ we get
\begin{align*}
	\sigma_i(\lambda,d) = \lambda_i \frac{\gamma_i(0)}{\sigma(\lambda)} \left(d + \frac2{\gamma_i(0)} \sum_{k=1}^{d-1} (d-k) \gamma_i(k)\right) / \delta(d),
\end{align*}
which gives the form of the factor $\delta(i,d)$ for $i=1,\ldots,n$.
\end{proof}

\begin{proof}[Proof of Proposition~\ref{prop:VAR.AR.1} and Corollary~\ref{cor:VAR.AR.1}]
Considering Proposition~\ref{prop:VAR.AR.1} we use that by assumption $(X_t(\lambda))_{t \in \za}$ is an $\AR{p}$-process as in~\eqref{eq:VAR.AR.1.ARP} and the identity $\lambda \X_t = X_t$:
\begin{align*}
	\lambda^T \X_t &= X_t(\lambda) \\
	\Leftrightarrow    \sum_{k=1}^p \lambda^T \Phi_k \X_{t-k} + \lambda^T\Z_t &= \sum_{k=1}^p \phi_k \lambda^T \X_{t-k} + \lambda^T\Z_t \\
		\Leftrightarrow  \lambda^T \sum_{k=1}^p (\Phi_k - \phi_k I_n)\X_{t-k}& = 0. 
\end{align*}
%Together with~\eqref{prop:VAR.AR.1.item2} in Proposition~\ref{prop:VAR.AR.1} this is equivalent to
%\begin{align*}
%	 \sum_{k=1}^p \lambda^T \Phi_k \X_{t-k} &=  \sum_{k=1}^p \phi_k \lambda^T \X_{t-k} \\
%	 	\Leftrightarrow  \lambda^T \sum_{k=1}^p (\Phi_k - \phi_k I_n)\X_{t-k}& = 0. 
%\end{align*}
As $(\X_t)_{t \in \za}$ takes values in $\re^n$ this forces 
\begin{align*}
	\lambda^T (\Phi_k - \phi_k I_n) = 0\quad\text{for } k= 1,\ldots,p.
\end{align*}	
To prove Corollary~\ref{cor:VAR.AR.1} consider that the above condition is true for arbitrary $\lambda \in \re^n$ if and only if
the rank of the matrix $\Phi_k - \phi_k I_n$ is zero for $k= 1,\ldots,p$ which concludes the proof.
\end{proof}

%\begin{proof}[Proof of Proposition~\ref{prop:ARVAR}]
%First assume that $\Phi_j = \phi_j I_n$	for $j=1,\ldots,p$. Then we have
%\begin{align*}
%	X_t &= \lambda^T \X_t  = \sum_{j=1}^p \lambda^T \Phi_j \X_{t-j} + \Z_t \\
%											 & = \sum_{j=1}^p \phi_j 	\lambda^T\X_{t-j}	+ \lambda^T \Z_t \\
%											 & = \sum_{j=1}^p \phi_j  X_{t-j} + \lambda^T \Z_t,											 
%\end{align*}
%which is an equation of the desired form.
%
%As a counter-example consider the equations
%\begin{align*}
%	X_t = a X_{t-1} + Z_t \quad\text{and}\\
%	\X_t =
%   \begin{pmatrix}
%     a & b \\
%     0 & a
%   \end{pmatrix}
%\X_{t-1} + \Z_t.
%\end{align*}
%Consider that for $|a|<1$ we get the following unvariate expression
%\[
%	X_t^2 = \sum_{j=1}^\infty a^j Z^2_{t-j}
%\]
%for the second component.
%
%Then one easily sees that for $\lambda = (\lambda_1,\lambda_2)^T$
%\begin{align*}
%	\lambda^T \X_t = \lambda_1 a X_{t-1}^1 + \lambda_2 a X_{t-1}^2 + \lambda_2 b X_{t-1}^2 +\lambda^T\Z_t  = a X_{t-1} + \lambda_2 b \sum_{j=1}^\infty a^{j-1} Z^2_{t-j-1} +\lambda^T\Z_t,
%\end{align*}
%which is inconsistent with the assumed equation for the portfolio return.
%\end{proof}
%\addcontentsline{toc}{section}{References}
%\bibliographystyle{siam}
%\bibliographystyle{alpha}
\bibliographystyle{plainnat}
\bibliography{quant}
\end{document}